\documentclass{amsart}

\usepackage{amsfonts}
\usepackage{amssymb}
\usepackage{amsmath, amsthm}
\usepackage{graphicx}
\usepackage[all]{xy}
\usepackage{color}
\usepackage{xcolor}
\usepackage{bm}
\usepackage{float}
\usepackage{enumitem}
\usepackage{tabularx}
\usepackage{multirow}
\usepackage{multicol}
\usepackage{array}
\newcolumntype{Z}{>{\centering\let\newline\\\arraybackslash\hspace{0pt}}X}

\usepackage{csquotes}
\usepackage{float}
\usepackage{subcaption}

\DeclareMathOperator{\proj}{proj}

\definecolor{e-mail}{rgb}{0,.40,.80}
\definecolor{reference}{rgb}{.20,.60,.22}
\definecolor{citation}{rgb}{0,.40,.80}

\definecolor{burntorange}{rgb}{0.8, 0.33, 0.0}
\definecolor{darkspringgreen}{rgb}{0.13, 0.55, 0.35}

\RequirePackage{verbatim}

\usepackage[colorlinks, linkcolor=reference,
 citecolor=citation,
         urlcolor=e-mail]{hyperref}

\RequirePackage{verbatim}

\RequirePackage{hyperref}

\RequirePackage{color}

\definecolor{conditional}{rgb}{0,1,0}
\definecolor{e-mail}{rgb}{0,.40,.80}
\definecolor{reference}{rgb}{.20,.60,.22}
\definecolor{mrnumber}{rgb}{.80,.40,0}
\definecolor{citation}{rgb}{.20,.60,.22}

\usepackage{bbm}

\usepackage{tikz}
\usetikzlibrary{arrows,decorations.markings}
\usepackage{url}

\newtheorem{theorem}{Theorem}
\newtheorem{proposition}{Proposition}

\newtheorem{lemma}{Lemma}
\newtheorem{corollary}{Corollary}

\theoremstyle{definition}

\newtheorem{remark}{Remark}
\newtheorem{definition}{Definition}
\newtheorem{example}{Example}

\DeclareMathOperator{\init}{in}

\DeclareMathOperator{\Wr}{Wr}

\DeclareMathOperator{\lead}{\text{lead}}

\setlist[enumerate]{leftmargin=.25in}
\setlist[itemize]{leftmargin=.25in}

\title{Algebraic identifiability of partial differential equation models}

\author[Byrne]{Helen M. Byrne}
\thanks{Byrne and Harrington are members of the Centre for Topological Data Analysis, which is funded by the EPSRC grant `New Approaches to Data Science: Application Driven Topological Data Analysis’ {\texttt{EP/R018472/1}}. This work was also partially supported by the NSF grants CCF-2212460, DMS-1760448, DMS-1853650, and DMS-1853482 and PSC-CUNY grant \#65605-00 53, a Royal Society University Research Fellowship, Renewal and Enhancement grant. For the purpose of Open Access, the authors have applied a CC BY public copyright licence to any  Author Accepted Manuscript (AAM) version arising from this submission.}
\address{Mathematical Institute,
Radcliffe Observatory Quarter,
Woodstock Road,
Oxford OX2 6GG,
UK and\\
Ludwig Institute for Cancer Research, Oxford, 
Nuffield Department of Clinical Medicine, 
Old Road Campus Research Building,
off Roosevelt Drive, Headington, Oxford, OX3 7DQ, UK (helen.byrne@maths.ox.ac.uk)} 
\author[Harrington]{
Heather A. Harrington}
\thanks{Harrington is supported by a Royal Society University Research Fellowship, URF Renewal, URF Enhancement grant and EPSRC \texttt{EP/R018472/1}.} 
\address{Mathematical Institute,
Radcliffe Observatory Quarter,
Woodstock Road,
Oxford OX2 6GG,
UK and\\
Faculty of Mathematics, Technische Universitat Dresden, 01062 Dresden, Germany and\\
Centre for Systems Biology Dresden (CSBD), 
Pfotenhauerstrasse 108, 
01307 Dresden, Germany and\\
Max Planck Institute of Molecular Cell Biology and Genetics (MPI-CBG), 
Pfotenhauerstrasse 108, 
01307 Dresden, Germany (harrington@mpi-cbg.de)} 
\author[Ovchinnikov]{
Alexey Ovchinnikov}
\thanks{The authors are grateful for the stimulating research environment at MPI-MiS Leipzig. The work started at the MPI-MiS Workshop on Differential Algebra.}
\address{Department of Mathematics, CUNY Queens College and Ph.D. Programs in Mathematics and Computer Science, CUNY Graduate Center, New York, USA (aovchinnikov@qc.cuny.edu)} 
\author[Pogudin]{Gleb Pogudin}\address{LIX, CNRS, \'Ecole Polytechnique, Institute Polytechnique de Paris, 1 rue Honor\'e d'Estienne d'Orves, 91120, Palaiseau, France (gleb.pogudin@polytechnique.edu)}
\author[Rahkooy]{Hamid Rahkooy}\address{Mathematical Institute,
Radcliffe Observatory Quarter,
Woodstock Road,
Oxford OX2 6GG,
UK (rahkooy@maths.ox.ac.uk)}
\author[Soto]{Pedro Soto}\address{Wellcome Centre for Human Genetics,
 Roosevelt Drive, Oxford OX3 7BN
UK (soto@maths.ox.ac.uk)}

\usepackage{footnote}
\usepackage{amsmath}
\usepackage{amssymb}
\usepackage{graphicx}
\usepackage{xcolor}
\usepackage{physics}
\usepackage{pdfpages}
\usepackage{algorithm}
\usepackage{algorithmic}
\usepackage{physics}

\newcommand{\bv}{\mathbf{v}}
\newcommand{\bw}{\mathbf{w}}
\newcommand{\by}{\mathbf{y}} 
\newcommand{\bk}{\mathbf{k}}  
\newcommand{\bg}{\mathbf{g}}  
\newcommand{\bff}{\mathbf{f}}

\newcommand{\RR}{\mathbb{R} }

\newcommand{\CC}{\mathbb{C} }

\usepackage{biblatex} 
\date{}
\begin{document}
\maketitle

\begin{abstract}
 Differential equation models are crucial to scientific processes. The values of model parameters are important for analyzing the behaviour of solutions. A parameter is called globally identifiable if its value can be uniquely determined from the input and output functions. To determine if a parameter estimation problem is well-posed for a given model, one must check if the model parameters are globally identifiable. This problem has been intensively studied for ordinary differential equation models, with theory and several efficient algorithms and software packages developed. A comprehensive theory of algebraic identifiability for PDEs has hitherto not been developed due to the complexity of initial and boundary conditions.  Here, we provide theory and algorithms, based on differential algebra, for testing identifiability of polynomial PDE models. We showcase this approach on PDE models arising in the sciences.
\end{abstract}

\section{Introduction}

Differential equations form the bedrock of numerous scientific and engineering models, particularly in the realms of biological and chemical interactions. Indeed, systems of ordinary and partial differential equations are integral to our understanding of such scientific phenomena. These models invariably incorporate time-varying dependent variables, input functions, and output functions, along with parameter vectors. 
Without loss of generality, parameters are positive scalars, which are independent of time, and whose values are often unknown. Parameter estimation or inference is determining the unknown parameter values from observations. Identifiability is a necessary condition for well-posed parameter estimation and inference.

There are different notions of parameter identifiability for differential equation models~\cite{Hong-GlobalIO2018}.  A parameter is called \emph{globally structurally identifiable} if its value can be uniquely determined from the input and output functions.
If a parameter takes finitely many values, it is called locally identifiable. A parameter that is neither globally, nor
 locally identifiable, is called unidentifiable. In this paper, we only
 consider global identifiability and, for brevity,  use the notation identifiable for
 global identifiability.
Given a model with 
observables, the \emph{parameter identifiability problem} is determining whether the parameters are identifiable.

For ordinary differential equation (ODE) models, the parameter identifiability problem has been extensively studied. Various approaches have emerged
in areas ranging from
control theory and dynamical systems (e.g.,  Taylor series approximations)
to
computational algebraic geometry and differential algebra.
Starting from Ritt~\cite{Ritt-book}, several theoretical results and algorithms have been developed within the last decades. 
The Rosenfeld-Gr\"obner algorithm and Gr\"obner bases are the two crucial concepts that underpin existing identifiability algorithms.
Based on these,  several software have been designed to test identifiability, e.g., DAISY~\cite{DAISY}, SIAN~\cite{SIAN}, 
COMBOS~\cite{COMBOS},  Structural Identifiability Toolbox~\cite{Ilmer2021}, and StructuralIdentifiability.jl~\cite{structidjl}.
An increase in the availability of spatio-temporal data enables the investigation of parameter values in PDE models. 

A key aim of this work is to adapt and extend the successful algebraic approaches for studying ODE identifiability to spatial systems. PDEs are inherently more complex than ODEs. 
There are additional challenges associated with the boundary conditions. 
With the exception of the work by \cite{Zhu2018,Renardy2022,ciocanel2023parameter,browning2023structural},  
there is no systematic study of PDE identifiability, to the best of our knowledge. 
Authors in~\cite{Zhu2018} investigate the identifiability and estimation of parameters of a chikungunya epidemic transmission model.
Structural identifiability of age-structured PDE models using a
differential algebra framework has been studied In~\cite{Renardy2022}.
Structural and practical identifiability of PDE models of
fluorescence recovery after photobleaching has been studied in~\cite{ciocanel2023parameter}.
In~\cite{browning2023structural}, the authors present a differential algebra approach to 
structural identifiability analysis on partially observed linear reaction-advection-diffusion PDE models.

We extend the identifiability problem to spatio-temporal models 
\begin{equation}\label{eq:sigma}
\Sigma =
  \begin{cases}
    \partial_t \bv=\bff\left(\bk, \bw, \bv, \partial_{x} \bv, \ldots, \partial_{x}^h \bv\right) \\
    \by=\bg\left(\bk, \bw, \bv, \partial_{x} \bv, \ldots,
      \partial_{x}^h \bv\right),
  \end{cases}
\end{equation}
which have broad applications in applied mathematics. 
Here $\bw(x,t)$, $\bv(x,t)$, $\by(x,t)$, and $\bk$ are vectors of inputs, state variables, outputs, and constant parameters, respectively and $\bff$ and $\bg$ are vectors of rational functions.
We present an algebraic approach for the PDE identifiability problem, focusing on models arising in applied mathematics. 
Our main results state that for a given model of the form System~\eqref{eq:sigma}, one can construct certain differential polynomials equations called Input Output (IO) equations such that the identifiability of the parameters can be obtained from the identifiability of the coefficients of the IO-equations (Theorem~\ref{thm:identifiability}), and that the coefficients of the IO-equations are identifiable if their Wronskian is nonsingular (Proposition~\ref{prop:identicbc}).
Based on our results, we present two algorithms (for IO-identifiability and strong identifiability), with implementation in {\sc Maple} for our illustrating examples, that provide a sufficient condition for solving the PDE identifiability problem. 
The main steps of our algorithms are the following:
\begin{itemize}
\item calculating input-output (IO) equations of the PDE model (Step 1 in Algorithms~\ref{alg:algebraic} and~\ref{alg:strong}), 
\item checking if the parameters can be uniquely determined from the coefficients of the IO-equations (Step 3 in Algorithm~\ref{alg:algebraic}, Steps 2 \& 3 in Algorithm~\ref{alg:strong}), and 
\item verifying that the coefficients of the IO-equations are themselves identifiable (Step 3 in Algorithm~\ref{alg:algebraic}, Step 4 in Algorithm~\ref{alg:strong}).
\end{itemize}

  In the ODE case, there is a key subtlety: even for generic initial conditions, the coefficients of the IO-equations are not always identifiable, see e.g. \cite[Example~2.14]{Hong-GlobalIO2018}.  Surprisingly, this subtlety does not occur in PDEs with generic initial/boundary conditions, as we prove in Theorem~\ref{thm:identifiability} and exploit in Algorithm~\ref{alg:algebraic}.
  
 Initial and boundary conditions in practical examples are not necessarily generic (cf. the ODE case \cite{SADA2003,GR2021}), which adds an additional layer of difficulty. To account for these conditions, the last step of our method finds potential linear dependencies between the monomials present in the input-output equations corresponding to the PDE models. The linear dependencies  are  then tested by computing the Wronskian of the monomials   and using differential algebra tools to 
determine if the Wronskian is non-singular. 
Using the Rosenfeld-Gr\"obner algorithm~\cite{Boulier2009computing}, we compute the normal form of the determinant of the Wronskian with respect to the differential ideal 
 of the model.
We then use the initial and boundary conditions of the model in order to refute the 
 vanishing of the determinant (see Proposition~\ref{prop:identicbc}). 

We demonstrate these results in 
Algorithm~\ref{alg:strong}, for testing the identifiability of standard models arising in applied mathematics. We consider different types of PDEs (parabolic, elliptic and hyperbolic), with a particular focus on parabolic PDEs that arise in mathematical biology.
We show that a scalar reaction-diffusion equation, Fisher's equation,
the coupled reaction-diffusion equations system, and a
reaction-diffusion system~\cite{murray-book1,gatenby1996reaction} are
all identifiable. We also demonstrate the wider applicability of this
framework on Laplace's equation (elliptic) and the wave equation
(hyperbolic).
    In Example~\ref{ex:single-output-LV}, we show that following our symbolic-computation based algorithm directly could be too demanding on the computational resources. We demonstrate how numeric computation with random values of parameters 
  gives evidence for identifiability at a generic point in reasonable computing time.

The organization of the paper is as follows. Section~\ref{sec:prelimanaries} outlines the preliminaries on differential algebra, detailing precise definitions and required results. Section~\ref{sec:wronskian} presents our results, offering the theoretical foundation for our PDE identifiability procedure using the Wronskian and generalizing the literature on ODE identifiability. Section~\ref{sec:algorithm} provides our two identifiability algorithms. In Section~\ref{sec:computations}, we showcase Algorithm~\ref{alg:strong} on the above suite of PDE models arising in mathematical biology.

\section{Differential equations to  differential algebra}\label{sec:prelimanaries}

We start the preliminaries by recalling differential polynomials, which provide a general framework for polynomial PDEs.

\begin{definition}[Ring of differential polynomials]
\begin{enumerate}
  \item[]
  \item A {\em differential ring} $(R,\Delta)$ is a commutative ring with a set $\Delta = \{\partial_1,\ldots,\partial_m\}$ of pairwise-commuting derivations $\partial_i:R\to R$, that is, maps such that, for all $a,b\in R$, $\partial_i(a+b) = \partial_i(a) + \partial_i(b)$ and $\partial_i(ab) = 
  \partial_i(a)b + a \partial_i(b)$. 
  \item A differential ring that is a field is called a {\em differential field}. 
  \item For a differential field $K$, the {\em ring of differential polynomials} in the variables $x_1,\ldots,x_n$ over a differential field $K$ is the polynomial ring in infinitely many variables 
  \[
    K[\partial_1^{n_1}\cdots\partial_m^{n_m}x_j\mid n_i\geqslant 0,\, 1\leqslant j\leqslant n]
  \] 
  with the derivations 
 extended from $K$ by
  \[\partial_i\left(\partial_1^{n_1}\cdots\partial_i^{n_i}\cdots\partial_m^{n_m}(x_j)\right) := \partial_1^{n_1}\cdots\partial_i^{n_i+1}\cdots\partial_m^{n_m}(x_j).
  \] 
  This differential ring is denoted by $K\{x_1,\ldots,x_n\}$.  
\end{enumerate}  
\end{definition}

\begin{definition}[Strong identifiability] \label{def:idinbc} 
Let $\mathbb{K}$ be one of the fields $\mathbb{R}$ and $\mathbb{C}$.
We fix a domain $\mathcal{D}$ in $\mathbb{K}^m$ on which a PDE system will be defined.
We will consider a PDE system in $n$ variables $\bv = (v_1, \ldots, v_n)$ and, for each $1 \leqslant i \leqslant n$, we fix a class $\mathcal{C}_i$ of function in $\mathcal{D}$ where the solutions for $v
_i$ will be sought.
The requirements on functions from $\mathcal{C}_i$ may involve, for example, regularity conditions (e.g., twice or infinitely differentiable) or boundary conditions on $\partial\mathcal{D}$.

Now we consider
a system of PDEs in $\mathcal{D}$ of the form
\begin{equation}\label{eq:PDEarb}
\mathbf{F}\left(\bk, \bv\right) = 0,
\end{equation}
where $\mathbf{F} = (F_1, \ldots, F_s)$ and $F_1, \ldots, F_s \in \mathbb{K}(\bk)\{\bv\}$ are differential polynomials in $\bv$ with coefficients that are rational functions in the scalar parameters $\bk=(k_1,\ldots,k_\ell)$. Let us fix a domain $\Omega \subset \mathbb{K}^\ell$, which will play a role of domain for the parameters (in our examples, if $\Omega$ is not specified, we will assume that $\Omega = \mathbb{K}^\ell$).
Furthermore, we fix outputs $\by = (y_1, \ldots, y_r)$ defined by given formulas
\[
y_i = \frac{G_{i}(\bk, \bv)}{Q(\bk, \bv)},\quad 1 \leqslant i \leqslant r,
\]
where $G_1, \ldots, G_r, H \in \mathbb{K}(\bk)\{\bv\}$. 
We will say that a rational function $h(\bk) \in \mathbb{K}(\bk)$ is \emph{identifiable} if, for all $\bk_1, \bk_2 \in \Omega$ and all $\bv_1, \bv_2 \in \mathcal{C}_1\times \ldots \times \mathcal{C}_n$ such that $\mathbf{F}(\bk_1, \bv_1) = \mathbf{F}(\bk_2, \bv_2) = 0$, we have
\[
  \mathbf{y}(\bk_1, \bv_1) = \mathbf{y}(\bk_2, \bv_2) \implies h(\bk_1) = h(\bk_2).
\]
\end{definition}

\begin{remark}
While the definition above is very general and stated in natural analytic terms, some of its properties make it challenging to use it in practice:
First, since it allows for arbitrary function classes and arbitrary polynomial PDEs, a complete constructive approach to verifying this property seems to be out of reach at the moment. 
Nevertheless, we will show that strong identifiability can be established in a variety of practical cases by a uniform approach (see Section~\ref{sec:computations}).

Second, a system that is ``almost always'' identifiable can become nonidentifiable according to this definition because of some special degenerate cases.
For example, a parameter $k$ in the ODE model $x'(t) = kx(t)$ with output $y(t) = x(t)$ will be considered nonidentifiable because the zero solution $x(t) = 0$ does not distinguish between different parameter values.
However, as long as $x(t)$ is nonzero, the value of $k$ is uniquely determined.
The standard mitigation of this issue in the ODE case is to restrict the discussion to generic solutions. 
The problem is that because of significantly more involved conditions on the existence and uniqueness of solutions of PDEs, the topology (and thus the notion of genericity) of the solution space of a PDE system may be much more involved even if the function classes $\mathcal{C}_1, \ldots, \mathcal{C}_n$ are infinitely differentiable functions subject to some boundary conditions.

Therefore, to obtain more refined results in a more restricted context, we will give another definition, which is a direct analogue of the notion of algebraic identifiability in the ODE case~\cite[Section~2.2]{ANSTETTCOLLIN2020139}. 
Note that, in the ODE case, this definition is equivalent to the analytic one~\cite[Proposition~3.4]{Hong-GlobalIO2018}, and we expect a similar equivalence result to hold in the PDE case (perhaps in the class of power series or analytic solutions).
\end{remark}

We will recall some relevant notions from differential algebra.

\begin{definition}[Differential ideals]
    An ideal $I$ of a differential ring $(R,\Delta)$ is called a {\em differential ideal} if, for all $a \in I$ and $\partial\in\Delta$, $\partial(a)\in I$. For $F\subset R$, the smallest differential ideal containing the set $F$ is denoted by $[F]$.

    For an ideal $I$ and set $S$ in a ring $R$, we denote $S^\infty$ to be the multiplicatively closed subset of $R$ generated by $S$ and \[I \colon S^\infty = \{r \in R \mid \exists a \in S^\infty\,\colon a r \in I\}.\]
  The set $I:S^\infty$ is also an ideal in $R$. If $S = \{a\}$ for some $a \in R$, we also denote $I:S^\infty$ by $I:a^\infty$. 
\end{definition}

As mentioned in the introduction, we will mostly focus on evolutionary PDEs.
More precisely, we will have one distinguished derivation $\partial_t$ with respect to the time and one or several spatial derivations $\partial_1, \ldots, \partial_m$, and we will consider systems of the form
\begin{equation}\label{eq:sigma_def}
\Sigma =
  \begin{cases}
    \partial_t \bv=\frac{\mathbf{f}(\bk, \bv, \bw)}{Q(\bk, \bv, \bw)} \\
    \by=\frac{\mathbf{g}(\bk, \bv, \bw)}{Q(\bk, \bv, \bw)},
  \end{cases}
\end{equation}
where $\bv$ and $\bw$ are the state and input variables of the model, $\by$ are outputs, $\bk$ are scalar parameters, and $\mathbf{f}, \mathbf{g}, Q$ are differential polynomials in $\bv, \bw$ with coefficients in $\mathbb{C}(\bk)$ of order zero with respect to $\partial_t$.
Note that we have focused on
complex numbers to unlock powerful algebraic tools we will use, and identifiability over complex numbers implies identifiability over reals.
Now we will give a formal definition of what generic solution of~\eqref{eq:sigma_def} is and then give our second definition of identifiability.

\begin{definition}[Generic solution]\label{def:generic_solution}
    Given $\Sigma$ as in~\eqref{eq:sigma_def}, we define the differential ideal of $\Sigma$ as 
    \[
    I_\Sigma=[Q\partial_t\bv - \bff, Q\by - \bg]:Q^\infty \subset \CC(\bk)\{\bv,\by,\bw\}.
    \] 
    In the same way as \cite[Lemma 3.1 and 3.2]{Hong-GlobalIO2018}, one can
     show that $I_\Sigma$ is a prime ideal and
     \begin{equation}
     I_\Sigma \cap \CC(\bk) \{ \bw \}\left[\partial_1^{i_1}\ldots \partial_m^{i_m}\bv,\ i_1, \ldots, i_m\geqslant 0\right]=\{0\} \label{eq:diff_ideal}.
\end{equation}
Since $I_\Sigma$ is prime, we can consider the field of fractions of $R / I_\Sigma$,
which we denote by $\mathcal{F}$. We will denote the images of
$\bv, \by, \bw$ in $\mathcal{F}$ by
$\hat{\bv}, \hat{\by}, \hat{\bw}$, respectively.
We will call them {\em the generic solutions}
of~\eqref{eq:sigma_def}. 
\end{definition}

\begin{definition}[Identifiability]\label{def:identifiability}
In the notation of the previous definition, we say that a parameter
$k_i \in \bk$ (or, more generally, a rational function of
parameters) is {\em identifiable} if
\begin{equation*}
  k_i \in \CC\langle \hat{\by},  \hat{\bw} \rangle,
\end{equation*}
where $\CC\langle \hat{\by},  \hat{\bw} \rangle$ is the smallest field extension of $\CC$ containing $\hat{\by},  \hat{\bw}$ and their derivatives.
\end{definition}

\section{Results}\label{sec:wronskian}
\subsection{Identifiability for generic solutions}

We will prove that, unlike the ODE case, all of the identifiable functions of parameters can be read off from particular relations between the input and output variables called input-output equations. 
Below we recall some necessary notions and constructions from constructive differential algebra.

\begin{definition}[Differential rankings and characteristic sets]
\begin{enumerate}
\item[] 
 \item
  A {\em differential ranking} on $K\{x_1,\ldots,x_n\}$ is a total order $>$ on \[X := \{\partial_1^{n_1}\cdots\partial_n^{n_m}(x_j)\mid n_i\geqslant 0,\, 1\leqslant j\leqslant n\}\] satisfying:
  \begin{itemize}
    \item for all $x \in X$ and $\partial\in\Delta$, $\partial(x) > x$ and
    \item for all $x, y \in X$ and $\partial\in\Delta$, if $x >y$, then $\partial(x) > \partial(y)$.
  \end{itemize}
  It can be shown that a differential ranking on $K\{x_1,\ldots,x_n\}$ is always well-ordered.

\item
  For  $f \in K\{x_1,\ldots,x_n\} \backslash K$ and differential ranking $>$,
\begin{itemize}
    \item 
    $\lead(f)$ is the element of $X$ appearing in $f$ that is maximal with respect to $>$.
    \item 
    The \emph{leading coefficient} of $f$ considered as a polynomial in $\lead(f)$ is denoted by $\init(f)$ and called the \emph{initial} of $f$. 
    \item
    The \emph{separant} of $f$ is $\frac{\partial f}{\partial\lead(f)}$, the partial derivative of $f$ with respect to $\lead(f)$.
    \item 
    The \emph{rank} of $f$ is $\rank(f) = \lead(f)^{\deg_{\lead(f)}f}$.
    \item 
    For $S \subset K\{x_1,\ldots,x_n\} \backslash K$, the set of \emph{initials} and \emph{separants} of $S$ is denoted by $H_S$.
    \item 
    For $g \in K\{x_1,\ldots,x_n\} \backslash K$, say that $f < g$ if $\lead(f) < \lead(g)$ or $\lead(f) = \lead(g)$ and $\deg_{\lead(f)}f < \deg_{\lead(g)}g$.
\end{itemize}
    \item For $f, g \in K\{x_1,\ldots,x_n\} \backslash K$, $f$ is said to be reduced w.r.t. $g$ if no proper derivative of $\lead(g)$ appears in $f$ and $\deg_{\lead(g)}f <\deg_{\lead(g)}g$.
    \item 
    A subset $\mathcal{A}\subset K\{x_1,\ldots,x_n\} \backslash K$
    is called {\em autoreduced} if, for all $p \in \mathcal{A}$, $p$ is reduced w.r.t. every  element of $\mathcal A\setminus \{p\}$. 
    One can show that every autoreduced set is always finite.
    \item Let $\mathcal{A} = \{A_1, \ldots, A_r\}$ and $\mathcal{B} = \{B_1, \ldots, B_s\}$ be autoreduced sets such that $A_1 < \ldots < A_r$ and $B_1 < \ldots < B_s$. 
    We say that $\mathcal{A} < \mathcal{B}$ if
    \begin{itemize}
      \item $r > s$ and $\rank(A_i)=\rank(B_i)$, $1\leqslant i\leqslant s$, or
      \item there exists $q$ such that $\rank(A_q) <\rank(B_q)$ and, for all $i$, $1\leqslant i< q$, $\rank(A_i)=\rank(B_i)$.
    \end{itemize}
    \item An autoreduced subset of the smallest rank of a differential ideal $I\subset K\{x_1,\ldots,x_n\}$
    is called a {\em characteristic set} of $I$. One can show that every non-zero differential ideal in $K\{x_1,\ldots,x_n\}$ has a characteristic set. Note that a characteristic set does not necessarily generate the ideal.
    \item For elements $r_1,\ldots,r_n$ is a differential ring $(R,\Delta)$, the Wronskian matrix $\Wr(r_1,\ldots,r_n)$ with respect to a given $\partial\in \Delta$ is the matrix
    \[
    \begin{pmatrix}
    r_1 & \ldots&r_n\\
    \partial r_1 &\ldots &\partial r_n\\
    \vdots & \ddots &\vdots\\
    \partial^{n-1} r_1 &\ldots &\partial^{n-1} r_n\end{pmatrix}.
    \]
\end{enumerate}
\end{definition}

\begin{definition}[IO-equations]\label{def:IO-equations}
Given a differential ranking on the differential variables $\by$ and $\bw$, the \emph{IO-equations} are defined as the monic characteristic presentation of the prime differential ideal $I_\Sigma \cap \CC(\bk)\{\by,\bw\}$ with respect to this ranking  (see \cite[Definition~6 and Section~5.2]{Ovchinnikov-Identifiability2023} for more details). For each differential ranking, such a monic characteristic presentation is unique~\cite[Theorem~3]{Boulier2000}.
\end{definition}

Now we are ready to state the main result of this section, namely, that all the identifiable functions can be read off the coefficients of the IO-equations. 
Interestingly, the corresponding statement is not true for ODEs \cite[Example~2]{Ovchinnikov-Identifiability2023}.

\begin{theorem}\emph{(cf. \cite[Theorems 1 \& 2]{Ovchinnikov-Identifiability2023})}\label{thm:identifiability}
For a model $\Sigma$ of the form~\eqref{eq:sigma_def}, the identifiable functions in $\CC(\bk)$ form a subfield, and this subfield is generated by the coefficients of any set of IO-equations of the model.
\end{theorem}

\begin{remark}\label{rmk:injectivity}   
Note that using Theorem~\ref{thm:identifiability}, 
one obtains identifiability of the coefficients $C$ of the IO-equations $F_j(\bw,\by)$, $1\leqslant j\leqslant s$, considered as differential polynomials in $\bw,\by$. 
Elements of $C$ are in fact the generators of the field of identifiable functions. 
  For each parameter $k_i$ in $\bk$, one can check its identifiability by verifying whether $k_i \in  \mathbb{Q}(C)$. 
  This can be reduced to the ideal membership problem as described in~\cite[Section 1.3]{Rational} (for a randomized version, see~\cite[Theorem~3.3]{structidjl}).
  The {\sc Maple} code~\cite{allident_code} has been implemented in~\cite{Ovchinnikov-AllIdentifiable21} for testing this ideal membership.
\end{remark}

Before proving Theorem~\ref{thm:identifiability}, we will proceed with some preparation.
Consider a PDE model $\Sigma$ as in~\eqref{eq:sigma}, its prime differential ideal
$I_\Sigma$, and the corresponding field of fractions $\mathcal{F}$ (see Definition~\ref{def:generic_solution}). 
Since $I_\Sigma$ is stable under $\partial_t$ and $\partial_{x}$,
the derivations $\partial_t$ and $\partial_{x}$ can be transferred
to $\mathcal{F}$ in a natural way.
\begin{definition}We will call an element $c \in \mathcal{F}$ a \emph{constant}
if $\partial_t c=\partial_{x} c=0$. 
\end{definition}
The key difference with the ODE
case will be the following lemma.

\begin{lemma}\label{lem:const} 
If $c \in \mathcal{F}$ and $\partial_x(c) = 0$, then $c\in\CC(\bk)$. In particular, the set of constants of $\mathcal{F}$ is $\CC(\bk)$.
\end{lemma}
\begin{proof}
  Let $c=\frac{p}{q} \in \mathcal{F}$ be  such $\partial_x(c) = 0$,  $p, q \in R=\CC(\bk)\{\bv,\by,\bw\}$ are
  coprime and at least one of them does not belong to
  $\CC(\bk)$.
  There exists a positive integer $N$ such that,
  using equations~\eqref{eq:sigma} and their derivatives,
  we can replace $Q^Np$ and $Q^Nq$ by elements of
  $\CC(\bk) \{ \bw \}[\partial_x^i\bv,\ i\geqslant 0]$ equivalent to $Q^Np$ and $Q^Nq$ modulo $I_\Sigma$ and, thus, yielding
  the same element of $\mathcal{F}$. By~\eqref{eq:diff_ideal}, $\mathcal{F}$ has a subfield isomorphic to the field of rational functions
  \[\mathcal{F}_0:=\CC(\bk)\langle \bw \rangle(\partial_x^i\bv,\ i\geqslant 0),\] and $\frac{p}{q}$ belongs to this subfield. With
  respect to $\partial_{x}$, the field $\mathcal{F}_0$ is isomorphic to the field of
  differential rational functions over $\bk$ in infinitely many
  variables
  $\bv, \bw, \partial_t \bw, \partial_t^2
  \bw, \ldots$. Since the constants of a field of differential
  rational functions are exactly the constants of the ground field, we
  deduce that $\frac{p}{q} \in \CC(\bk)$.
\end{proof}

\begin{lemma}\label{lem:fieldofdef} Let $\Delta$ be a finite set of derivations, $L \subset K$  differential fields, and $X$ a finite set of variables. 
Let $P$ be a prime non-zero differential ideal of $K\{X\}$ such that the ideal generated by $P$ in $\overline{K}\{X\}$ is prime, where $\overline{K}$ is the algebraic closure of $K$. 
If $\mathcal{C}$ is a monic characteristic presentation of $P$, then the field of definition of $P$ over $L$ is the field extension of $L$ generated by the coefficients~of~$\mathcal{C}$.
\end{lemma}
\begin{proof}
Mutatis mutandis proof of \cite[Proposition~2]{Ovchinnikov-Identifiability2023}.
\end{proof}

\begin{corollary}[{cf. \cite[Corollary~1]{Ovchinnikov-Identifiability2023}}]\label{cor:generateIO}
If $\mathcal{C}$ is a monic characteristic presentation of the prime differential ideal $J := I_\Sigma \cap \mathbb{C}(\bk)\{\by,\bw\}$, then field of definition of $J$ over $\mathbb{C}$ is generated over $\mathbb{C}$ by the coefficients of $\mathcal{C}$. 
\end{corollary}
\begin{proof}
This follows from Lemma~\ref{lem:fieldofdef} because the differential ideal generated by $J$ in $\overline{\mathbb{C}(\bk)}\{\by,\bw\}$ is prime.
\end{proof}

\begin{lemma}[{cf. \cite[Lemma 1]{Ovchinnikov-Identifiability2023}}]\label{lem:support}
Consider a polynomial $P \in I_\Sigma \cap \CC(\bk)\{\by, \bw\}$ with at least one of the coefficients being one. If there is no element in  $I_\Sigma \cap \CC(\bk)\{ \by,\bw\}$ whose support is a subset of the support of $P$, then all coefficients of $P$ are identifiable.
\end{lemma}
\begin{proof}
We write $P=m_0+\sum_{i=1}^{\ell} c_i m_i$, where $m_0, \ldots, m_{\ell}$ are monomials in 
\[
\left\{\partial_t^{i} \partial_{x}^j \by, \partial_t^i \partial_{x}^j \bw \mid i, j \geqslant 0\right\}
\] 
and $c_1, \ldots, c_{\ell} \in \CC(\bk)$. 
    Let $(\hat\bv,\hat\by,\hat\bw)$ be the generic solution of $I_\Sigma$. 
    We consider the Wronskian \[W(\hat{\by}, \hat{\bw}):= \Wr(m_1(\hat{\by}, \hat{\bw}), \ldots, m_{\ell}(\hat{\by}, \hat{\bw}))\]  with respect to $\partial_{x}$. 
    If this Wronskian was singular, then, by \cite[Theorem~3.7]{Kaplansky},
    there would exist $a_1, \ldots, a_{\ell} \in F$ constant with   respect to $\partial_{x}$ so that \[a_1 m_1(\hat{\by}, \hat{\bw})+\ldots+a_{\ell} m_{\ell}(\hat{\by}, \hat{\bw})=0.\] By Lemma~\ref{lem:const},  $a_1, \ldots, a_{\ell} \in \CC(\bk)$. 
    This yields a nonzero element \[a_1 m_1+\ldots+a_{\ell} m_{\ell} \in I_\Sigma\] with the support being a  subset of the support of $P$.
    Thus, $W$ is nonsingular. 
    By taking the derivatives of $P(\hat{\by}, \hat{\bw})=0$ of orders from $0$ to $\ell - 1$ with respect to $\partial_x$, we obtain: 
    $$
    \left(\begin{array}{c}
        m_0(\hat{\by}, \hat{\bw}) \\
        \partial_{x} m_0(\hat{\by}, \hat{\bw}) \\
        \vdots \\
        \partial_{x}^{\ell-1} m_0(\hat{\by}, \hat{\bw})
      \end{array}\right)
    =W(\hat{\by}, \hat{\bw}) 
    \left(\begin{array}{c}
    c_1 \\
    c_2 \\
    \vdots \\
    c_{\ell}
    \end{array}\right).
    $$ 
    Viewing this as a nonsingular linear system in $c_1, \ldots, c_{\ell}$ over $\CC\langle\hat{\by},  \hat{\bw}\rangle$ and applying Cramer's rule, we deduce that each of $c_1, \ldots, c_{\ell}$ belongs to $\CC\langle\hat{\by}, \hat{\bw} \rangle$.
\end{proof}

\begin{corollary}\label{cor:identifiability}
  If system~\eqref{eq:sigma} is of the form
  \[
  \Sigma =
  \begin{cases}
    \partial_t \bv=\bff\left(\bk, \bw, \bv, \partial_{x} \bv, \ldots, \partial_{x}^h \bv\right) \\
    \by=\bv,
  \end{cases}
  \]
  that is, all  states are observable, then all  coefficients of
  $\bff$ as polynomials in
  $\bw, \bv, \partial_{x} \bv, \ldots, \partial_{x}^h \bv$ are
  identifiable.
\end{corollary}

\begin{proof} It is sufficient to prove that polynomials
  $\partial_t \bv-\bff$ satisfy the condition of Lemma~\ref{lem:support}. Assume that one of them, say
  $\partial_t v_1-f_1$ does not. 
  Therefore, there is a polynomial
  $g \in I_\Sigma$ with the support being a proper subset of the support of
  $\partial_t v_1-f_1$. 
  By cancelling $\partial_t v_1$, we obtain a
  $\CC(\bk)$-linear combination $\tilde{g}$ of $\partial_t v_1-f_1$
  and $g$ with the support being a subset of the support of
  $f_1$. Then $\tilde{g}$ is a nonzero element of
  $I \cap \CC(\bk) \{ \bw \}[\partial_x^i\bv,\ i\geqslant 0]$, contradicting
  \eqref{eq:diff_ideal}.
\end{proof}

We are now ready to prove our main result, Theorem~\ref{thm:identifiability}.

\begin{proof}[Proof of Theorem~\ref{thm:identifiability}]
Let 
\[S =\{c \in \CC(\bk)\mid c \text{ is identifiable}\}= \CC(\bk)\cap\CC\langle \hat{\by},\hat{\bw}\rangle,\]
and so $S$ is a subfield in $\CC(\bk)$. Let $\mathcal{C}$ be a set of IO-equations of $\Sigma$ and $C$  the set of coefficients of $\mathcal{C}$. We will prove that $S = \CC(C)$.

We first show that $S \subset \CC(C)$.
    For this, let $c \in \CC(\bk)$ be identifiable.
    We will show how the proof of~\cite[Theorem 1]{Ovchinnikov-Identifiability2023} extends to PDEs to prove that $c \in \CC(C)$. By Corollary~\ref{cor:generateIO}, the field of definition of $I_\Sigma\cap\CC(\bk)\{\by,\bw\}$ is equal to $\CC(C)$. In particular, $I_\Sigma\cap\CC(\bk)\{\by,\bw\}$ is generated over $\CC(\bk)$ by $I_\Sigma\cap\CC(C)\{\by,\bw\}$.
    
    Since $c \in \CC\langle \hat{\by},\hat{\bw}\rangle$, 
     there exist $g \in \CC\{\by,\bw\}\setminus I_\Sigma$ and $h \in \CC\{\by,\bw\}$ such that $gc + h \in I_\Sigma$. Therefore, there exist $m_1,\ldots,m_r \in \CC(\bk)\{\by,\bw\}$ and $p_1,\ldots,p_r \in I_\Sigma \cap \CC(C)\{\by,\bw\}$ such that   
     \[
     gc+h = m_1p_1+\ldots+m_rp_r.
     \] 
     Now assume that $c \not\in \CC(C)$. By \cite[Theorem 9.29, p. 117]{milneFT}, there exists an automorphism $\sigma$ on $\overline{\CC(\bk)}$ that fixes $\CC(C)$ pointwise, but does not fix $c$, i.e., $\sigma(c) \neq c$.   Extend $\sigma$ to $\overline{\CC(\bk)}\{\bv,\by,\bw\}$ by letting $\sigma$ fix $\bv ,\by$, and $\bw$.  We have in $\overline{\CC(\bk)}\{\bv,\by,\bw\}$ that
    \begin{equation}
    \label{eq:quot2}
    \begin{aligned}
        \left(gc+h\right) - \sigma(gc+h) & = \left( m_1p_1+\ldots+m_rp_r \right) - \sigma(m_1p_1+\ldots+m_rp_r) \\
        g(c-\sigma(c)) & = (m_1-\sigma(m_1))p_1+\ldots+(m_r-\sigma(m_r))p_r.
    \end{aligned}  
    \end{equation}
    Let $\overline{I_\Sigma }$ denote the differential ideal generated by $I_\Sigma$ in $\overline{\CC(\bk)}\{\bv,\by,\bw\}$. Since $\overline{I_\Sigma }$ is a prime differential ideal and the right-hand side of \eqref{eq:quot2} belongs to $\overline{I_\Sigma }$, it follows that either $g \in \overline{I_\Sigma }$ or $c-\sigma(c) \in \overline{I_\Sigma }$.  As $\sigma(c) \ne c$, we have that $c-\sigma(c)$ is a non-zero element of $\overline{\CC(\bk)}$. Since $\overline{I_\Sigma }$ is a proper ideal, therefore $c-\sigma(c) \not\in \overline{I_\Sigma }$.  Therefore, $g \in \overline{I_\Sigma }$.  Hence, \[g \in \overline{I_\Sigma } \cap \CC(\bk)\{\bv,\by,\bw\} = I_\Sigma,\] contradicting the assumption on $g \not\in I_\Sigma$.
        
     For the proof of the converse, note that we have shown in Lemma~\ref{lem:const} that the set of constants of $\mathcal{F}$ coincides with $\CC(\bk)$. Therefore, the assumption of  \cite[Theorem~2]{Ovchinnikov-Identifiability2023} holds. So we follow the proof of~\cite[Theorem~2]{Ovchinnikov-Identifiability2023} for PDEs.

    Let $J := I_\Sigma \cap \CC(\bk) \{\by, \bw\}$.  Let         $\mathcal{B}$ be the set of differential monomials in $\by$ and $\bw$ indexed by $\mathbb{N}$ such that the indexing respects a ranking. Then $\mathcal{B}$ is a basis for $\CC(\bk)\{\by,\bw\}$ as a vector space over $\CC(\bk)$. 
    Consider a basis $\mathcal{B}_J$ for $J$, as a $\CC(\bk)$-subspace of $\CC(\bk)\{\by,\bw\}$. Each element of $\mathcal{B}_J$ is a finite linear combination of the element of $\mathcal{B}$. 
    Autoreduce $\mathcal{B}_J$ as a set of vectors expressed in the basis $\mathcal{B}$ and call it $\mathcal{B}_J'$. 
    Note that $\mathcal{B}_J'$ is still a basis of $J$ as a vector space. 
    The field of definition of $J$ over $\CC$ is contained in the field generated by the coefficient of $\mathcal{B}_J'$ written as linear combinations of $\mathcal{B}$. 
    Therefore, it is sufficient to prove that these coefficients are identifiable.
        
    Let $P \in J$ be the differential polynomial corresponding to a member of $\mathcal{B}_J'$ of $J$. 
    Let $q$ be a polynomial whose support is a subset of $P$ and whose monomials are linearly dependent modulo $J$ over $\CC(\bk)$. The representation of  $q$ in basis $\mathcal{B}$ can be reduced to zero by $\mathcal{B}_J$, however, $q$ cannot be reduced to zero by $P$. Also, $q$ reduced by $P$ cannot be further reduced by other elements of $\mathcal{B}_J'$, as $\mathcal{B}_J'$ is a reduced basis. This proves that there is no $q$ whose support is a subset of the support of $P$ and whose support is linearly dependent. Therefore, by Lemma~\ref{lem:support}, the coefficients of $P$ are identifiable.   
\end{proof}

\subsection{Identifiability for solutions with specified initial and boundary conditions}
In this section, we are in the setup of Definition~\ref{def:idinbc}.
Consider a PDE system of the form~\eqref{eq:PDEarb} and each $\mathcal{C}_i$  contained in the $C^\infty$-functions on a domain $\mathcal{D}$. 
Fix a differential ranking such that any derivative of $\bv$ is greater than any derivative of~$\by$, and let
\begin{equation}\label{eq:dec}
I_\Sigma = [C_1]:H_{C_1}^\infty\cap\ldots\cap [C_r]:H_{C_r}^\infty
\end{equation}
be a decomposition computed by~\cite[Section~4]{Boulier2000} or~\cite[Algorithm~7.2]{Hubert2} with respect to the differential ranking. 
For each $i$, $1\leqslant i\leqslant r$, let \[\widetilde{C_i} = C_i\cap \mathbb{C}(\bk)\{\by\}.\]
If the system~$\Sigma$ is as in~\eqref{eq:sigma}, then, as we discussed, $I_\Sigma$ is a prime differential ideal and, as a result, one can remove redundant components in~\eqref{eq:dec} and obtain $r = 1$ and $\widetilde{C_1} \subset C_1 \subset I_\Sigma$, see \cite[Theorem~3.2.1]{KR}. If $\Sigma$ is more general and $I_\Sigma$ is not necessarily a prime differential ideal, consider the set $\widetilde{C} = \prod_{i=1}^r \widetilde{C_i} \subset I_\Sigma$.
By dividing every element of $\widetilde{C}$ by an element of $\mathbb{C}(\bk)$, for every $c \in \widetilde{C}$,
we pick a representation of the form
\begin{equation}\label{eq:grouping}
c = p_{0,c} + \sum_{j=1}^{q_c}a_{j,c}\cdot p_{j,c},
\end{equation}
where $p_0,\ldots,p_{q_c} \in \mathbb{C}\{\by\}$ and $a_1\ldots,a_{q_c} \in \mathbb{C}(\bk)$. Denote
$\Wr_c = \Wr(p_{1,c},\ldots, p_{q_c,c})$.
\begin{remark}\label{rem:grouping}
Such a grouping~\eqref{eq:grouping} is implemented in {\sc Maple} in~\cite{allident_code}, the function {\tt DecomposePolynomial}.
\end{remark}

\begin{proposition}
\label{prop:identicbc}
Let $c \in \widetilde{C}$, defined above. If, for all  
 $\widehat{\bk} \in \Omega$ and
 solutions  $(\widehat{\bv},\widehat{\by})$ of~\eqref{eq:PDEarb} with the parameter values $\widehat{\bk}$,
there exists a point $(t_0,x_0) \in \mathcal{D}$ such that the matrix  $\Wr_c\left(\widehat{\by}\right)(t_0,x_0)$ with respect to $\partial_x$ or $\partial_t$ is invertible, then 
\begin{itemize}\item
all $a_{1,c},\ldots,a_{q_c,c}$ are identifiable.\item Moreover, let $k_{i_1},\ldots, k_{i_c}$ be the parameters that explicitly appear in $a_{1,c},\ldots,a_{q_c,c}$. If the ``coefficient map''
\[\varphi_c: (k_{i_1},\ldots, k_{i_c})\mapsto (a_{1,c},\ldots,a_{q_c,c})\] is injective, then the parameters $k_{i_1},\ldots, k_{i_c}$ are identifiable.
\end{itemize}
\end{proposition} 
\begin{proof}
Suppose that there exists $j$ such that $a_{j,c}$ is not identifiable according to Definition~\ref{def:idinbc}. Let $\widehat{\bk}_1, \widehat{\bk}_2 \in \Omega$ and solutions $(\widehat{\bv}_1,\widehat{\by}_1)$ and  $(\widehat{\bv}_2,\widehat{\by}_2)$ of~\eqref{eq:PDEarb}  be such that
\begin{equation}
\label{eq:nonident}
\widehat{\by}_1 = \widehat{\by}_2\quad \text{but}\quad a_{j,c}(\widehat{\bk}_1)\ne a_{j,c}(\widehat{\bk}_2).
\end{equation}
Consider the square system of linear equations
\begin{equation}\label{eq:Wraj}
\begin{cases}
\sum\limits_{j=1}^{q_c}a_{j,c}\cdot p_{j,c} =-p_{0,c} \\
\sum\limits_{j=1}^{q_c}a_{j,c}\cdot \partial(p_{j,c}) =-\partial(p_{0,c})\\
\qquad\vdots\\
\sum\limits_{j=1}^{q_c}a_{j,c}\cdot \partial^{q_c-1}(p_{j,c}) =-\partial^{q_c-1}(p_{0,c})
\end{cases}
\end{equation}
in the unknowns $a_{1,c},\ldots,a_{q_c,c}$, whose matrix is $\Wr_c(\by)$. Assume that, for all $\widetilde{\bk} \in \Omega$ and solutions $(\widetilde{\bv},\widetilde{\by})$ of~\eqref{eq:PDEarb} with parameters $\widetilde{\bk}$, there exists $(t_0,x_0) \in \mathcal{D}$ such that  the matrix  $\Wr_c\left(\widetilde{\by}\right)(t_0,x_0)$ with respect to $\partial_x$ ($\partial_t$ is considered similarly) is invertible. Substituting $\widehat{\by}_1$ or $\widehat{\by}_2$ and the corresponding point $(t_0,x_0)$  into~\eqref{eq:Wraj} and solving returns $\left(a_{1,c}(\widehat{\bk}_1),\ldots,a_{q_c,c}(\widehat{\bk}_1)\right)$ or $\left(a_{1,c}(\widehat{\bk}_2),\ldots,a_{q_c,c}(\widehat{\bk}_2)\right)$, respectively. By the first part of~\eqref{eq:nonident}, these tuples are equal, which contradicts with the second part of~\eqref{eq:nonident}.

For the proof of the second item, suppose that the coefficient map $\varphi_c$ is injective. If $k_{i_j}$ is not identifiable, then there exist $\widehat{\bk}_1, \widehat{\bk}_2 \in \Omega$ and solutions $(\widehat{\bv}_1,\widehat{\by}_1)$ and $(\widehat{\bv}_2,\widehat{\by}_2)$ of~\eqref{eq:PDEarb} with  parameter values $\widehat{\bk}_1$ and $\widehat{\bk}_2$ such that
\[
\widehat{\by}_1 = \widehat{\by}_2\quad \text{but}\quad \proj_{k_{i_j}}(\widehat{\bk}_1)\ne \proj_{k_{i_j}}(\widehat{\bk}_2).
\]
Hence, by the injectivity of $\varphi_c$, there exists $\ell$ such that $a_{\ell,c}(\widehat{\bk}_1)\ne a_{\ell,c}(\widehat{\bk}_2)$, which contradicts the identifiability of $a_{\ell,c}$ established above.
\end{proof}

\section{Algorithms}\label{sec:algorithm}

The correctness of Algorithms~\ref{alg:algebraic} and~\ref{alg:strong}, which we present in this section, follows from Theorem~\ref{thm:identifiability} and Proposition~\ref{prop:identicbc}, respectively.

\begin{algorithm}[!h]
\caption{Assessing identifiability in the sense of Defintion~\ref{def:identifiability}}\label{alg:algebraic}
\begin{description}
    \item[Input] A rational PDE system of the form
    \[
    \begin{cases}
      \partial_t{\bv}=\bff \left( \bk, \bw, \bv, \partial_{x}{\bv}, \ldots, \partial_{x}^h \bv \right) \\
      \by=\bg\left( \bk, \bw, \bv, \partial_{x}{\bv}, \ldots, \partial_{x}^h \bv\right).
    \end{cases}
    \]

    \item[Output] A list of generators of the field of identifiable functions of the system and, for each parameter, $k_i$ whether it is identifiable or not.
\end{description}

\begin{enumerate}[label = \textbf{(Step~\arabic*)}, leftmargin=*, align=left, labelsep=2pt, itemsep=0pt]
    \item Using
  \texttt{RosenfeldGroebner} with an elimination ranking $\bv > \by,\bw$,  eliminate the variables $\bv$ and obtain a set $\mathcal{S}$ of IO-equations of the input PDE system.
  \item Let $\mathcal{C} \subset \mathbb{C}(\bk)$ be the set of the coefficients of $\mathcal{S}$.
  \item \textbf{Return}:
  \begin{itemize}
      \item For the generators of the field of identifiable functions, the set $\mathcal{S}$. It can be additionally simplified using~\cite{allident_code} by applying {\tt FieldToIdeal} and then {\tt FilterGenerators} (or~\cite[Section~5]{Meshkat2009});
      \item For each $k_i \in \bk$, the result of the test $k_i \in \mathbb{C}(\bk)$ (using, e.g., \cite[Section 1.3]{Rational}, see~\cite{allident_code} for implementation).
  \end{itemize}
\end{enumerate}
\end{algorithm}

\begin{algorithm}[!h]
\caption{Approach to establishing strong identifiability  (Defintion~\ref{def:idinbc})}\label{alg:strong}
\begin{description}
    \item[Input] 
    A system of rational PDEs of the form
    \[
    \mathbf{F}\left(\bk, \bv\right) = 0, \quad \mathbf{y} = \mathbf{G}(\bk, \bv).
    \]
    together with some set of requirements $\mathcal{R}$ (regularity, boundary conditions, etc.) on the states $\bv$ and a domain $\Omega$ for the parameter values.

    \item[Output] For each parameter $k_i$, either returns that it is strongly identifiable, or that the test was inconclusive.
\end{description}

\begin{enumerate}[label = \textbf{(Step~\arabic*)}, leftmargin=*, align=left, labelsep=2pt, itemsep=0pt]
    \item Using
  \texttt{RosenfeldGroebner} with an elimination ranking $\bv > \by$,  eliminate the variables $\bv$, and obtain a set $\mathcal{S} = \{S_1, \ldots, S_\ell\}$ of IO-equations of the input PDE system.
  We denote the set of their coefficients by $\mathcal{C}$.
  \item For each $j$, $1 \leqslant j\leqslant \ell$, write $S_j$ in the form 
  \[
    F_j = p_{0,j} + \sum_{i=1}^{L_j}a_{i, j}\cdot p_{i, j},
  \]
  with $a_{i, j} \in \mathbb{C}(\bk)$ and $p_{i, j} \in \mathbb{C}\{\by\}$ (see Remark~\ref{rem:grouping}).
  Compute the Wronskians $\operatorname{Wr}_{j,x}$ and $\operatorname{Wr}_{j,t}$ of $p_{1, j}, \ldots, p_{L_j, j}$ with respect to $\partial_x$ and $\partial_t$, respectively.
  \item For each $j$, $1 \leqslant j\leqslant \ell$,
compute the normal forms $N_{j,x}$ and  $N_{j,t}$  of $\operatorname{Wr}_{j,x}$ and  $\operatorname{Wr}_{j,t}$, respectively, w.r.t. the input PDE system,
and obtain a sufficient condition 
  $\mathcal{I} = \bigwedge_j(N_{j,x}\ne 0 \lor N_{j,t}\ne 0)$ for the identifiability of $\mathcal{C}$.
  \item Check if $\mathcal{I}$ holds under the requirements $\mathcal{R}$ and for the parameters in $\Omega$\footnote{Since there are almost no assumptions on the requirements, this step is non-algorithmic. We show, how to do this in practice in the next section.}.
  \item If $\mathcal{I}$ does not hold, \textbf{return} ``inconclusive'' for every parameter.
  \item If $\mathcal{I}$ does hold, \textbf{return} ``identifiable'' for every parameter $k_i$ such that $k_i \in \mathbb{C}(\bk)$ (can be checked with, e.g., \cite[Section 1.3]{Rational}, see~\cite{allident_code} for implementation), and ``inconclusive'' for the rest of the parameters.
\end{enumerate}
\end{algorithm}

\section{Examples of PDE identifiability}\label{sec:computations}

In this section, we consider several models arising in mathematical biology (Section~\ref{subsec:math-bio-pde}), as well as other natural phenomena (Section~\ref{subsec:applied-maths-pdes}). We use Proposition~\ref{prop:identicbc} and Algorithm~\ref{alg:strong},  to test the identifiability of the parameters in these models. 
\subsection{PDEs arising in mathematical biology}\label{subsec:math-bio-pde}
We study four well-known, and increasingly complex, PDE systems in mathematical biology \cite{murray-book1}. 
All of these examples use parabolic PDEs.
\begin{example}[Scalar Reaction--Diffusion Equation]\label{subsubsec:nutrient}

We start by considering the following reaction--diffusion equation in which $c$ represents the concentration of a diffusible nutrient, such as oxygen or glucose, which is consumed at a rate which is an increasing, saturating function of its concentration~\cite{murray-book1}:
\begin{equation}\label{eq:nutrient}
  \partial_tc(x,t) =d\partial_x^2 c(x,t) +\frac{\lambda c(x,t)}{c_0+c(x,t)},
\end{equation}
where the set of parameters is $\{d, c_0, \lambda\}$, and the boundary
conditions are given by:
\begin{align}
  \partial_x c(x,t)(0,t) = & 0 \label{eq:nutrient-bc1}\\
  c(R,t) = & 1 \label{eq:nutrient-bc2}\\
  c(x,0) = & 1, \quad 0 \leq x \leq R. \label{eq:nutrient-bc3}
\end{align}

Following Definition~\ref{def:idinbc}, our field is $\mathbb{K} = \RR$ and $\Omega = \RR_+^3$.
We follow Algorithm~\ref{alg:strong}  to check
the identifiability of the parameters of~\eqref{eq:nutrient}. Details of
the computations can be found in the {\sc Maple} worksheet \texttt{nutrient.mpl}\footnote{\url{https://github.com/rahkooy/PDE-Identifiability}}.

Considering the numerator of the rational  function~\eqref{eq:nutrient}, 
i.e., the differential polynomial 
\begin{equation}\label{eq:nutrient-in-out}
    -d c(x,t)\partial_x^2 c(x,t) -d c_0\partial_x^2c(x,t)+ 
    c(x,t)\partial_tc(x,t) + c_0 \partial_t c(x,t) - \lambda c(x,t),
\end{equation}
we collect four monomials of~\eqref{eq:nutrient-in-out}
whose coefficient (with respect to the elements in $\CC(d,c_0,\lambda)$) is not 1, and compute $W$,  the determinant of their Wronskian to check if it is nonsingular. 
To do so, we compute the normal form of $W$ (using Rosenfeld-Gr\"obner),
and check when its coefficients are zero.

Note that if we obtain nonsingularity of the Wronskian from a subset of the coefficient, it implies that the Wronskian is nonsingular. The first ten coefficients are zero if and only if 
\begin{equation}
    \partial_tc(x,t) = 0 \quad \text{or} \quad \partial_x c(x,t)=0.
\end{equation}

So we investigate the following two cases. 

\paragraph{Case 1. $\partial_tc(x,t)=0$}
In this case, the function $c$ does not depend on $t$ anymore, and therefore,  $c(x,t)=c(x)$. Using the boundary conditions~\eqref{eq:nutrient-bc2} and~\eqref{eq:nutrient-bc3}, 
we obtain $c(x)=1$, a constant  function. On the other hand, 
considering $\partial_tc(x,t)=0$, the differential polynomial ~\eqref{eq:nutrient-in-out} will be simplified to the following 
polynomial (which is the LHS of  ODE, rather than the original PDE):
  \begin{equation}\label{eq:nutrient-odex}
    -d c(x)\partial_x^2 c(x) -d c_0\partial_x^2c(x) - \lambda c(x).
  \end{equation}

Substituting $c(x)=1$  into~\eqref{eq:nutrient-odex}, we obtain
$- \lambda$, which is not zero, as the parameters are assumed to be positive. So this case cannot happen.

\paragraph{Case 2. $\partial_xc(x,t)=0$}
This condition means that $c$ does not depend on $x$, hence, $c(x,t)=c(t)$. So substituting this condition into the PDE~\eqref{eq:nutrient}, we obtain the following ODE:
  \begin{equation}\label{eq:nutrient-odet}
    \partial_tc(t) = \frac{ \lambda c(t)}{c_0+c(t)}.
  \end{equation}
Using boundary condition~\eqref{eq:nutrient-bc2}, i.e., $c(R,t)=1$,
we obtain $\partial_t c(t) =0$. Therefore
ODE~\eqref{eq:nutrient-odet} is equal to zero if and only if
\begin{equation}
\frac{ \lambda c(t)}{c_0+c(t)} = \frac{ \lambda}{c_0+1}=0,
\end{equation}
which can happen only if $\lambda =0$. This is impossible according
to our assumption on the positivity of the parameters, i.e., $\Omega = \RR_+^3$.  So Case 2
neither can happen.

In conclusion, none of the cases considered above can happen. This means
that the Wronskian matrix is non-singular. So, by Proposition~\ref{prop:identicbc}, the coefficients of the
PDE~\eqref{eq:nutrient-in-out}, i.e., $\{d, d c_0, c_0, \lambda\}$,
are all identifiable. 
\end{example}  
\begin{example}[Fisher's Equation]
Next, we consider
Fisher's equation that describes the diffusive spread of a species which undergoes logistic growth~\cite{murray-book1}. 
Fisher's equation is given by the following PDE:
\begin{equation}\label{eq:fisher}
  \partial_tn(x,t) = d \partial_x^2 n(x,t) + r
  n(x,t)\left(1-\frac{n(x,t)}{k} \right),
\end{equation}
where $n(x,t)$ is the input function, $x,t$  are the variables, and
$d, r, k$ are the parameters. The boundary conditions are given by

\begin{equation}\label{eq:fisher-bc}
  n(x,t)  \to 
            \begin{cases}
              k, \quad x \to -\infty \\
              0, \quad x \to \infty,
            \end{cases} 
\end{equation}
and the initial condition is given by prescribed $u(x,0)$ as follows, which is compatible with the boundary  conditions~\eqref{eq:fisher-bc}.
\begin{equation}\label{eq:fisher-ic}
  n(x,0)  = n_0(x)  = \frac{ke^{-\alpha x}}{1+e^{-\alpha x}}.
\end{equation}

Following Definition~\ref{def:idinbc}, our field is $\mathbb{K} = \RR$ and $\Omega = \RR_+^3$.
Below we follow Algorithm~\ref{alg:strong} for checking the identifiability of the parameters. Details of the computations can be found in the {\sc Maple} worksheet \texttt{fisher.mpl}\footnote{\url{https://github.com/rahkooy/PDE-Identifiability}}.
Simplifying~\eqref{eq:fisher}, we obtain
  \begin{equation}\label{eq:fisher-simple}
  \partial_tn(x,t) - d \partial_x^2 n(x,t) - rn(x,t) + \frac{rn(x,t)^2}{k},
\end{equation}
from which we collect all of the monomials, except for $\partial_tn(x,t)$   whose coefficient is $1$, and compute $W$, the determinant of their Wronskian. 
    Using Rosenfeld-Gr\"obner, we compute the normal form of the determinant of $W$, which is a rational function. We would like to check if this rational function is identically zero. Equivalently, we would like to check if the numerator of this rational function, say $N$, is zero. For this, we consider the coefficients of $N$ as a polynomial in terms of parameters. There are four such coefficients (which are differential polynomials in terms of variables $x$ and $t$). 
    
Computing the coefficients of the normal form of $W$, one can see that the coefficients are zero if and only if 
\begin{equation}\label{eq:fisher-coeff-rg}
    \partial_xn(x,t)=0 \quad \text{or} \quad \partial_tn(x,t)=0,
\end{equation}
which leads to the following cases.
  
\paragraph{Case 1. $\partial_xn(x,t)=0$.}
In this case, the function $n(x,t)$ does not depend on $x$, hence  $n(x,t)=n(t)$. However, the initial condition~\eqref{eq:fisher-ic}   implies that $n$ is a function of $x$ unless $\alpha=0$, which cannot happen as the parameters are supposed to be positive.

\paragraph{Case 2. $\partial_tn(x,t)=0$.}
This case implies that $n$ does not depend on $t$, hence, $n(x,t)=n(x)$. From the initial condition~\eqref{eq:fisher-ic}, we have that 
\begin{equation}\label{eq:fisher-nx}
    n(x) = n_0(x) = \frac{k e^{-\alpha x}}{1 + e^{-\alpha x}}.
\end{equation}

We substitute $\partial_tn(x,t)=0$ into Fisher's Equation~\eqref{eq:fisher} to obtain the following ODE:
\begin{equation}\label{eq:fisher-odet}
    -d \partial_x^2 n(x) - rn(x) + \frac{rn(x)^2}{k} = 0.
\end{equation}
Equation~\eqref{eq:fisher-nx} must satisfy ODE~\eqref{eq:fisher-odet}. Evaluating the ODE at $n(x)$, we obtain
\begin{equation}\label{eq:fisher-evalnx}
    \frac{\left((\alpha^2d - r)e^{-\alpha x} - \alpha^2d -
        r\right)ke^{-\alpha x}}{\left(1 + e^{-\alpha x}\right)^3}.
\end{equation}
Equation~\eqref{eq:fisher-evalnx} is zero if and only if its numerator is zero, which is the case if and only if either $k=0$ (which is not  possible according to our assumption that parameters are positive), or
\begin{equation}
    (\alpha^2d-r) e^{-\alpha x}-\alpha^2 d-r =0.
\end{equation}
The above equation can happen only if the coefficient of the exponential term, as well as the constant term are zero. This leads to the equations
\begin{align}
    \alpha^2 d-r &=0 \\
    -\alpha^2 d-r & = 0,
\end{align}
which leads to $2 \alpha^2d=0$. But this cannot happen as $\alpha$  and $d$ are assumed to be non-zero.

So from the above discussions, we conclude that none of the above two cases can happen, hence, the coefficients of the Fisher  equation~\eqref{eq:fisher}, that is, by Proposition~\ref{prop:identicbc}, these functions of parameters are identifiable: $d$, $\frac{r}{k}$ and $\frac{1}{k}$. By the second part of Proposition~\ref{prop:identicbc}, the parameters $d,r,k$ are identifiable.
\end{example}
  
\begin{example}[Coupled Reaction-Diffusion Equations]\label{ex:lotka-volterra}
Next, we consider the following system of two coupled reaction--diffusion equations. Here species $u$ and $v$ undergo random motion and logistic growth while competing for resources~\cite{murray-book1}.
The PDEs are: 
\begin{equation}\label{eq:lv}
  \begin{cases}
    \partial_t u(x,t) =d_1\partial_x^2 u(x,t)
    +u(x, t)\left(a_1 - b_1  u(x, t) - c_1  v(x, t)\right),    \\
   \partial_t v(x,t) =d_2\partial_x^2 v(x,t) 
    +v(x,t)\left(a_2 - b_2 u(x, t) - c_2 v(x, t)\right) ,\\
    y_1(x,t) = u(x,t),\\
    y_2(x,t) = v(x,t).
  \end{cases}
\end{equation}
The model parameters are $a_1,b_1,c_1,d_1,a_2,b_2,c_2,d_2$. The boundary conditions are given by
\begin{align}
  u(x,t)  & \to 
            \begin{cases}
              \frac{a_1}{b_1} \quad x \to -\infty \\
              0 \quad x \to \infty 
            \end{cases} \label{eq:lv-bc1}\\
  v(x,t) & \to
           \begin{cases}
             \frac{a_2}{b_2} \quad x \to \infty \\
             0 \quad x \to -\infty
           \end{cases} \label{eq:lv-bc2}
\end{align}
and the initial conditions are given by prescribed $u(x,0)$ and
$v(x,0)$, e.g.,
\begin{align}
  u(x,0) & = u_0(x)  = \frac{\left(\frac{a_1}{b_1}\right) e^{-\alpha_1
           x}}{1+e^{-\alpha_1 x}} \label{eq:lv-ic1} \\
  v(x,0) & = v_0(x)  = \frac{\left(\frac{a_2}{b_2}\right) e^{-\alpha_2
           x}}{1+e^{-\alpha_2 x}}\cdot \label{eq:lv-ic2}
\end{align}
We note that for the special case with $a_1 < 0 = b_1 = b_2$ equations (\ref{eq:lv}) reduce to the classical Lotka-Volterra equations with random motion. 

Following Definition~\ref{def:idinbc}, our field is $\mathbb{K} = \RR$ and $\Omega = \RR_+^{8}$.
Below we follow our method~\eqref{alg:strong} for testing the identifiability of the parameters. Details of the computations can be found in the {\sc Maple} worksheet \texttt{LV-PDE.mpl}\footnote{\url{https://github.com/rahkooy/PDE-Identifiability}}.

For the first equation, the normal form of the determinant of the Wronskian yields a polynomial with 34 coefficients. 
We have considered 22 first coefficients and obtained that the determinant of the Wronskian is zero if and only if
  \begin{equation}\label{eq:lv-cf-decompose}
    \partial_x u(x,t)=0 \quad \text{or} \quad 
    \partial_t u(x,t)=\partial_t v(x, t)=0 \quad \text{or} \quad v(x, t)=0,
  \end{equation}

which leads us to the following cases.
\paragraph{Case 1. $\partial_x u(x,t)=0$}
In this case, the derivatives of $u(x,t)$ with respect to $x$ is zero, which means that the function $u$ does not
  depend on $x$, i.e., $u(x,t)=u(t)$. However,
  according to the initial condition~\eqref{eq:lv-ic1}, $u$  depends on $x$, unless
  $-\alpha_1=0$, in which case $u=\frac{a_1}{b_1}$ is a constant function. Using the boundary
  condition~\eqref{eq:lv-bc1} when
  $x \rightarrow -\infty$,  functions $u$ becomes
  zero. This implies that $\frac{a_1}{b_1} = 0$,
  which cannot happen as we have assumed that the parameters are
  nonzero. So this case does not happen either.

\paragraph{Case 2. $\partial_t u(x, t) = \partial_t v(x, t)=0$}
  In this case, $u$ and $v$ do not depend on it, i.e., $u(x,t)=u(x)$
  and $v(x,t)=v(x)$. Since the initial conditions~\eqref{eq:lv-ic1}
  and~\eqref{eq:lv-ic2} only depend on $x$, therefore,
  \begin{align}
    u(x,t) & = u(x) =  \frac{\left(\frac{a_1}{b_1}\right) e^{-\alpha_1
             x}}{1+e^{-\alpha_1 x}} \label{eq:lv-ux} \\
    v(x,t) & = v(x)  = \frac{\left(\frac{a_2}{b_2}\right) e^{-\alpha_2
             x}}{1+e^{-\alpha_2 x}}.\label{eq:lv-vx}
  \end{align}
  On the other hand, substituting
  $\partial_t u(x, t) = \partial_t v(x, t)=0$ into the original PDE
  system~\eqref{eq:lv}, one obtains the following two ODEs:
  \begin{align}
    & -d_1 \partial_x^2 u(x) - u(x)\left(a_1 - b_1 u(x) - c_1 v(x)\right) = 0 \label{eq:lv-ode1} \\
    & -d_2 \partial_x^2 v(x) - v(x)\left(a_2 - b_2 u(x) - c_2 v(x)\right) = 0. \label{eq:lv-ode2} 
  \end{align}
  Therefore, $u(x)$ and $v(x)$ should satisfy the ODEs
  ~\eqref{eq:lv-ode1} and~\eqref{eq:lv-ode2}. Having
  substituted~\eqref{eq:lv-ux} into the ODEs~\eqref{eq:lv-ode1}, one
  obtains
  \begin{equation}
    \frac{1}{b_1(1 + e^{-ax})^3}
    \left(
      a_1e^{-ax}
      \left(
        (d_1a^2 + 2c_1v(x) - a_1)e^{-ax} - d_1a^2 + v(x)e^{-2ax}c_1 + c_1v(x) - a_1
      \right)
    \right).
  \end{equation}
  Finally, substituting~\eqref{eq:lv-vx} into the above, we obtain a rational function of the form
  \begin{equation}\label{eq:lv-ode-plugged}
    \frac{1}{c_2(1 + e^{bx})b_1(1 + e^{-ax})^3} F(x)  a_1e^{-ax},
  \end{equation}
  where
  \begin{equation}
    \begin{aligned}
      F(x) = & ( -a_2c_1e^{-x(2a-b)} + \left( (-d_1a^2 + a_1)c_2 - 2a_2c_1
      \right)e^{x(-a + b)} \\
      & - c_2(d_1a^2 - a_1)e^{-ax} \\
      & + \left(c_2(d_1a^2 + a_1) - a_2c_1\right)e^{bx} \\
      & + c_2(d_1a^2 +a_1)
    \end{aligned}
  \end{equation}
  The rational function in~\eqref{eq:lv-ode-plugged} is identically
  zero if and only if either $a_1 (e)^{-a x}=0$ or $F(x)=0$. Since
  $a_1 \ne 0$, it cannot happen that $a_1 (e)^{-a x}=0$. The second
  case, $F(x)=0$, can happen if the exponential terms appearing in
  $F(x)$ are linearly dependent, which only can happen if the
  exponents of those terms are equal.  It is enough to consider the
  cases in which the first exponent is equal to some other
  exponents. This way, we obtain below four cases, each of which is
  discussed and it has been shown that none can happen.

  For convenience, set
  \begin{align}
    u &:=-a_2c_1, \\
    w &:=-d_1 a^2+a_1 \\
    t & := d_1 a^2+a_1.
  \end{align}
  Then the coefficients of $F(x)$ will be $u, wc_2-2 u, c_2w, c_2t-u$ and $c_2t$.
  Also, let $y=e^x$. The four cases are the following.
  \paragraph{Case 2.1. $-2a+b=-a+b$}
  This implies that $a=0$ which cannot happen.

  \paragraph{Case 2.2. $-2a+b=-a$}
  This means that $a=b$, and the exponents  become $y^{-a},y^0, y^{-a}, y^a, y^0$.
  collecting coefficients of equal exponents, we obtain that
  \begin{align}
    u+c_2w & =0, \\
    c_2w-2 u+c_2t & =0, \\
    c_2t-u & =0.
  \end{align}
  Simplifying the above we have $2u=0$ which is not possible.

  \paragraph{Case 2.3. $-2a+b=b$}
  This case implies that $a=0$ which is not possible.

  \paragraph{Case 2.4. $-2a+b=0$}
  From the conditions of this case, we obtain that $b=2a$, and the
  exponents will be $y^0, y^a, y^{-a}, y^{2a}, y^0$.  since $y^a$ and
  $y^{2a}$ are linearly independent, if their coefficients do not kill
  each other, then we are done. This is the case as the coefficients
  are $c_2w-2u$ and $c_2w$.

Considering coefficients of the normal form of the determinant of the second Wronskian, we obtain the following conditions:
  \begin{align}
    & \frac{\partial^2}{\partial_x \partial_t} u(x, t)= \partial_x v(x,t)=0 \quad \text{or} \label{eq:lv-cf-decompose2-1}\\ 
    & \partial_t u(x,t) = \partial_t v(x, t)=0 \quad \text{or} \label{eq:lv-cf-decompose2-2} \\
    & u(x, t)=0 \quad \text{or} \quad  \label{eq:lv-cf-decompose2-3} \\ 
    & v(x, t)=0. \label{eq:lv-cf-decompose2-4}
  \end{align}
Three out of the four cases in~\eqref{eq:lv-cf-decompose2-1}, \eqref{eq:lv-cf-decompose2-2}, \eqref{eq:lv-cf-decompose2-3}, and \eqref{eq:lv-cf-decompose2-4} have been considered earlier. The only remaining case $\frac{\partial^2}{\partial_x \partial_t} u(x, t)= \partial_x v(x,t)=0$. Similar to the argument for $\partial_x u(x,t)=0$, one can check that $\partial_x v(x,t)=0$ cannot happen. So this case is also impossible.

In conclusion, none of the Wronskians can be identically singular. Hence, by Proposition~\ref{prop:identicbc},
the coefficients of the PDEs~\eqref{eq:lv} are identifiable. Since the coefficients of the system are the parameters, 
all parameters are identifiable.
\end{example}


\begin{example}[Single Output Lotka Volterra]\label{ex:single-output-LV}
  In this example, we consider the single output version of the Coupled Reaction-Diffusion Equations studied in Example~\ref{ex:lotka-volterra}:
  \begin{equation}\label{eq:lv_single_out}
  \begin{cases}
    \partial_t u(x,t) =d_1\partial_x^2 u(x,t)
    +u(x, t)\left(a_1 - b_1  u(x, t) - c_1  v(x, t)\right),    \\
   \partial_t v(x,t) =d_2\partial_x^2 v(x,t) 
    +v(x,t)\left(a_2 - b_2 u(x, t) - c_2 v(x, t)\right) ,\\
    y(x,t) = u(x,t).
  \end{cases}
\end{equation}

  The initial and boundary conditions are the same as the Coupled Reaction-Diffusion Equations and are given by~\eqref{eq:lv-bc1},~\eqref{eq:lv-bc2},~\eqref{eq:lv-ic1}, and~\eqref{eq:lv-ic2}. Similarly, the ground field is $\mathbb{K} = \RR$ and $\Omega = \RR_+^{8}$.
  
  For this example, we initially tried to apply Algorithm~\ref{alg:strong}. We  used~\cite{allident_code} applying {\tt FieldToIdeal} and then {\tt FilterGenerators} to regroup the monomials and to simplify the system, and obtained a 13 $\times$ 13 matrix instead of the original 17 $\times$ 17 Wronskian. However, computing the determinant (and then normal form) of this simpler matrix was not possible within a reasonable time and memory. 
  Hence, instead of computing directly the normal form of the determinant of the Wronskian, we used the following steps to make computations faster: 
  \begin{enumerate} 
      \item First,  we compute the normal form of each entry of the Wronskian 
      \item Secondly, using the equations of the system~\eqref{eq:lv_single_out}, we eliminate the derivatives of $u(x,t)$ and $v(x,t)$ with respect to $t$, so that the equations only depend on $u, v$ and their derivatives with respect to $x$. Then we could substitute the initial conditions efficiently to the system (substituting initial conditions into the original system is very time-consuming).
      \item Third, we substituted random values for parameters as computations with symbolic parameters were not possible in a reasonable time. 
      \item Lastly, we computed the numeric value of the determinant of the Wronskian, evaluated in a generic point $x$. 
  \end{enumerate}

  While computing the determinant of the Wronskian and its normal form did not finish within three days, the above optimizations made the computation finish in less than four hours, most of which is spent on normal form computations. Having stored the output of the normal forms, one can carry on the computations for different values of $x$ and the parameters in the initial conditions in just less than a second.
  
  As the determinant was not zero in a generic point, we conclude that the determinant is not a zero function for random nonzero values of the parameters. Although this does not prove strong identifiability, however, we succeeded in presenting strong numeric evidence that generically the parameters are identifiable. 
  Additional studies can check particular values of the parameters that vanish the Wronskian.
  We note that due to the lack of resources for symbolic computations, we chose random values for parameters. For the details of the computations, we refer to our {\sc Maple} code~\texttt{LV-Single-Output.mpl}\footnote{\url{https://github.com/rahkooy/PDE-Identifiability}}.

Finally, we also see that the parameters $c_1$ and $c_2$ are not identifiable using the following argument. Following  
Definition~\ref{def:idinbc}, let 
\begin{gather*}
\bk_1=(a_1,b_1,c_1,d_1,a_2,b_2,c_2,d_2),\ \ \bv_1=(u,v)\\ 
\bk_2=(a_1,b_1,\lambda c_1,d_1,a_2,b_2,\lambda c_2,d_2),\ \ \bv_2=(u,v/\lambda),
\end{gather*}
where $\lambda$ is an arbitrary nonzero number. One sees that, substituting $\bk_2$ and $\bv_2$ into~\eqref{eq:lv_single_out}, $\lambda$ will be eliminated and we will obtain~\eqref{eq:lv_single_out} again. However, $\bk_1$ and $\bk_2$ are not equal. This means that $\bk_1$ and $\bk_2$ are not identifiable. As $\bk_1$ and $\bk_2$ only differ in the values of $c_1$ and $c_2$, one can conclude that $c_1$ and $c_2$ are not identifiable. 

\end{example}

\begin{example}[A Reaction-Diffusion Model of Cancer Invasion]

Our final example from mathematical biology is a PDE model of cancer
invasion which couples two reaction-diffusion equations for the concentrations of tumour cells $v(x,t)$ and acid (or pH) $w(x,t)$
with a time-dependent ordinary differential equation for healthy cells $u(x,t)$. The healthy cells and the cancer cells undergo logistic growth and compete with each other for space. The healthy cells do not move but they are killed by acid which is produced by the tumour cells and undergoes natural decay while it diffuses through the domain. The diffusive movement of the cancer cells is assumed to be a linearly decreasing function of the concentration of healthy cells~\cite{gatenby1996reaction}. 
The model is given by the following three PDEs:
\begin{align}
  &\partial_t u(x,t) = r_1 u(x,t) \left(1 - \frac{u(x,t)}{k_1} - \frac{v(x,t)}{k_2} a_{12}\right) 
    - d_1 w(x,t) u(x,t), \label{eq:diffusion1}\\
  &\partial_t v(x,t) = r_2 v(x,t) \left(1 - \frac{v(x,t)}{k_2} - \frac{u(x,t)}{k_1} a_{21}\right) 
    + d_2 \partial_x \left( (1-u(x,t)) \partial_x v(x,t)\right), \label{eq:diffusion2}\\
  &\partial_tw(x,t)  = d_4 \partial_x^2 w(x,t) + r_3 v(x,t) - d_3
    w(x,t), \label{eq:diffusion3}
\end{align}
along with the output equations
\begin{align}
  y_1(x, t)& = u(x, t), \\
  y_2(x, t) &= v(x, t), \\
  y_3(x, t) & = w(x, t).
\end{align}
The parameters of the system are
$r_1,k_1,a_{12},k_2,d_1,r_2,a_{21},d_2,d_3,r_3,d_4$, the boundary conditions are given by
\begin{align}
  u(x,t)  & \to 
            \begin{cases}
              k_1, \quad x \to +\infty \\
              0, \quad x \to -\infty 
            \end{cases} \label{eq:diffusion-bc1}\\
  v(x,t) & \to
           \begin{cases}
             0, \quad x \to +\infty \\
             k_2, \quad x \to -\infty
           \end{cases} \label{eq:diffusion-bc2} \\
  w(x,t) & \to
           \begin{cases}
             0, \quad x \to +\infty \\
             \frac{k_2r_3}{d_3}, \quad x \to -\infty,
           \end{cases} \label{eq:diffusion-bc3}
\end{align}
and the initial conditions are given by prescribed $u(x,0)$ and
$v(x,0)$, e.g.,
\begin{align}
  w(x,0) & = 0 \label{eq:diffusion-ic1}\\
  u(x,0) & = \frac{k_1 e^{\gamma_1 x}}{1+e^{\gamma_1 x}} \label{eq:diffusion-ic2} \\
  v(x,0) & = \frac{k_2 e^{-\gamma_2 x}}{1+e^{-\gamma_2 x}}\cdot \label{eq:diffusion-ic3}
\end{align}
Our ground field is $\mathbb{K} = \RR$ and $\Omega = \RR_+^{11}$, corresponding to the number of  parameters.

Details of the computations can be found in the {\sc Maple} worksheet \texttt{reaction-diffusion.mpl}\footnote{\url{https://github.com/rahkooy/PDE-Identifiability}}.
  
We consider each of the three equations separately, computing three  Wronskians and checking if they are nonsingular. We note that for the third equation, as there are no monomials with coefficient 1, we divide the polynomial by $k_1 k_2$ so that it has a monomial with coefficient one. Below is the summary of our computations.

\begin{itemize}
   \item The first Wronskian yields 40 coefficients. Considering ten of the coefficients, 
    Rosenfeld-Gr\"obner results in the following conditions:
    \begin{equation}
        \partial_x u(x,t)=0 \quad \text{or} \quad v(x,t)=0.
    \end{equation}
    The above conditions can easily be refuted by looking at the boundary and initial conditions, using arguments similar to the Coupled Reaction-Diffusion equations.
    \item For the second equation, the normal form contains 8 coefficients. Our procedure results in the following conditions for Wronskian to be singular:
    \begin{equation}
      w(x, t) \frac{\partial w(x, t)}{\partial t\partial x} = \pdv{w(x, t)}{t} \pdv{w(x, t)}{x} \quad \text{or} \quad 
      w(x, t)=0 \quad \text{or} \quad 
     v(x, t)=0.
    \end{equation}
    The second and third conditions, i.e., $w(x, t) =0$ and $v(x,t)=0$ can be easily refuted using boundary and initial conditions. For the first condition, we show that one can solve using the specified initial condition and obtain only zero solution. More precisely, substituting~\eqref{eq:diffusion-ic1} in the first condition, one obtains that either $\pdv{w(x,t)}{x} = 0$ at $t = 0$ or $\pdv{w(x, t)}{t} = 0$ at $t = 0$. One can easily check that both of these contradict the initial and boundary conditions, as the parameters are not allowed to be zero.
    \item For the third equation, the first ten coefficients result in the following conditions
    \begin{equation}
    \pdv{u(x,t)}{x}=0 \quad \text{or} \quad v(x,t)=0
    \end{equation}
    The second condition ($v(x,t)=0$) can be trivially refuted. For the first conditions, adding $\pdv{u(x,t)}{x}$ to the system and computing Rosenfeld-Gr\"obner, we obtain 9 equations, three of which are the input-output equations, three are the partial derivatives of the output equations with respect to $x$, and the remaining three equations involve the parameters. The latter three are of interest to us. They show that the states are not constant with respect to $x$. One can check that the boundary and initial conditions would imply that several parameters are zero, which is not permitted according to our assumptions. This can simply be seen as the normal forms of $u, v$, and $w$ with respect to the Rosenfeld-Gr\"obner are zero.
\end{itemize}
\end{example}

\subsection{PDEs from applied mathematics}\label{subsec:applied-maths-pdes}
Finally, in this section, we consider two well-known PDEs in applied mathematics that model natural phenomena. 

\begin{example}\label{ex:elliptic-pde}
The following PDE is an example of an elliptic PDE.
\begin{equation}\label{eq:elliptic-pde}
\partial_x^2 u(x,y) + \partial_y^2 u(x,y) = \theta,    
\end{equation}
where $0 \leq x \leq L$ and $0 \leq y \leq H$. The boundary conditions are 
\begin{align}
    u(x,H) & = \frac{\alpha x}{L} \label{eq:elliptic-bc1} \\
    u(L,y) & = \frac{\alpha y}{H},
\end{align}
and the initial conditions are 
\begin{align}
    u(x,0) & = 0 \\
    u(0,y) & = 0.
\end{align}
The parameters of the system are $\theta$ and $\alpha$, and the domain for parameters is $\Omega=\RR_+^{2}$.

In order to remove the parameter $\alpha$ from the boundary condition, one can define a new variable $v := u / \alpha$. Then the new equation and boundary conditions will be 
\begin{equation}\label{eq:elliptic-pde-reparameter}
\alpha \partial_x^2 v(x,y) + \alpha \partial_y^2 v(x,y) = \theta,    
\end{equation}
the boundary conditions will be 
\begin{align}
    v(x,H) & = \frac{x}{L} \\
    v(L,y) & = \frac{y}{H},
\end{align}
and the initial conditions will be 
\begin{align}
    \alpha v(x,0) & = 0 \\
    \alpha v(0,y) & = 0.
\end{align}
Since $\alpha$ is assumed to be non-zero, the initial conditions become 
\begin{align}
     v(x,0) & = 0 \\
     v(0,y) & = 0.
\end{align}
Dividing both sides of~\eqref{eq:elliptic-pde-reparameter} by $\alpha$, we obtain $\partial_x^2 v(x,y) + \partial_y^2 v(x,y) = \frac{\theta}{\alpha}$. Hence, $\frac{\theta}{\alpha}$ is identifiable, as it is in terms of derivatives of $v(x,y)$, however, the parameters $\alpha$ and $\theta$ are not identifiable. For the special case of $\alpha=1$, the parameter $\theta$ will be identifiable.   
\end{example}
\begin{remark}
    Note that the current work does not address the general case of systems with parameters in the boundary conditions. This is a potential future work using prolongations as in SIAN~\cite{SIAN}.
\end{remark}

\begin{example}\label{{ex:hyperbolic-pde}}
The following PDE is the well-known \emph{wave equation} and is an example of a hyperbolic PDE.
\begin{equation}\label{eq:wave-equations}
    \partial_t^2u(x,t) - \sigma^2 \partial_x^2 u(x,t) =0,
\end{equation}
where $\sigma$, the parameter, is the wave speed.
The initial conditions are given by 
\begin{align}\label{eq:ic-hyperbolic}
     u(x,0) = L  - x, \quad \partial_t u(x,0) = \beta.
\end{align}
The output function is $y(x,t) = u(x,t)$.
By d'Alembert's formula~\cite[Theorem~2.15]{Olver},~\eqref{eq:wave-equations} and~\eqref{eq:ic-hyperbolic} define a unique solution
\begin{equation}\label{eq:hyperbolic-solution}
    y(x,t) = -x+\beta t+L.
\end{equation}

Following Definition~\ref{def:idinbc}, our ground field is $\mathbb{K} = \RR$ and $\Omega = \RR_+$. Let us now try to apply Algorithm~\ref{alg:strong} (Proposition~\ref{prop:identicbc}) to see that computing Wronskian is essential. 
The IO-equation is
\begin{equation}
\label{eq:IOwave}\partial_t^2y(x,t) - \sigma^2 \partial_x^2 y(x,t) =0.
\end{equation}
Of its two monomials, the coefficient of $\sigma^2 \partial_x^2y(x,t)$ is not one. The Wronskian of this monomial (with respect to both $x$ and $t$) is itself. 
Furthermore, this monomial $\partial_x^2y(x, t)$ vanishes on every solution~\eqref{eq:hyperbolic-solution}, so Algorithm~\ref{alg:strong} is not applicable here.
Finanlly, since the solution~\eqref{eq:hyperbolic-solution} depends only on the initial conditions and not on $\sigma$, $\sigma$ is nonidentifiable by Definition~\ref{def:idinbc}.
\end{example}

\printbibliography

\end{document}